\documentclass[12pt]{article}

\usepackage{latexsym}
\usepackage{epsfig,graphics}
\usepackage{amsmath,amssymb,amsthm}
\usepackage{graphics,epsfig,calc}
\usepackage{upref}
 \newcommand{\beqn}{\begin{eqnarray}}
 \newcommand{\eeqn}{\end{eqnarray}}
 \newcommand{\be}{\begin{equation}}
 \newcommand{\ee}{\end{equation}}
 \newcommand{\ba}{\begin{array}}
 \newcommand{\ea}{\end{array}}
\newcommand{\br}{\begin{remark}}
 \newcommand{\er}{\end{remark}}
\newcommand{\bc}{\begin{cor}}
 \newcommand{\ec}{\end{cor}}

 \newcommand{\pa}{\partial}
 \newcommand{\re}{\ref}
 \newcommand{\ci}{\cite}
 \newcommand{\ds}{\displaystyle}
 \newcommand{\la}{\label}
 \newcommand{\rIm}{{\rm Im\5}}
 \newcommand{\rRe}{{\rm Re\5}}
 \newcommand{\supp}{{\rm supp~}}

\newcommand{\fr}{\frac}

\newcommand{\ov}{\overline}

\newcommand{\ti}{\tilde}

\newcommand{\cA}{{\cal A}}

\newcommand{\cG}{{\cal G}}
\newcommand{\cH}{{\cal H}}

\newcommand{\cL}{{\cal L}}

\newcommand{\cO}{{\cal O}}
\newcommand{\cP}{{\cal P}}
\newcommand{\cR}{{\cal R}}

\newcommand{\cU}{{\cal U}}
\newcommand{\cV}{{\cal V}}

\newcommand{\ve}{\varepsilon}

\newcommand{\de}{\delta}

\newcommand{\al}{\alpha}

\newcommand{\Ga}{\Gamma}
\newcommand{\si}{\sigma}
\newcommand{\om}{\omega}

\newcommand{\na}{\nabla}
\newcommand{\Si}{\Sigma}
\newcommand{\lam}{\lambda}

\newcommand{\5}{{\hspace{0.5mm}}}

\newcommand{\R}{\mathbb{R}}

\newcommand{\C}{\mathbb{C}}

\newtheorem{theorem}{Theorem}[section]

\newtheorem{lemma}[theorem]{Lemma}
\newtheorem{remark}[theorem]{Remark}

\newtheorem{cor}[theorem]{Corollary}
\newtheorem{pro}[theorem]{Proposition}
\newcommand{\bp}{\begin{pro}}
\newcommand{\ep}{\end{pro}}

\newcommand{\const}{\mathop{\rm const}\nolimits}
\newcommand{\sgn}{\mathop{\rm sgn}\nolimits}

\begin{document}
%%%%%%%%%%%%%%%%%%%%%%%%%%%%%%%%%%%%%%%%%%%%%%%%%%%%%%%%%%%%%%%%%%%%%%%%%%%%%%%%%%%%%
\begin{titlepage}
\bigskip\bigskip\bigskip

\begin{center}
{\Large\bf Weighted Energy Decay for 1D Dirac Equation}
\vspace{1cm}
\\
{\large E.~A.~Kopylova}
\footnote{Supported partly by the FWF, DFG and RFBR grants}
\\
{\it Institute for Information Transmission Problems RAS\\
B.Karetnyi 19, Moscow 101447,GSP-4, Russia}\\
e-mail:~elena.kopylova@univie.ac.at
\end{center}

\date{}

%%%%%%%%%%%%%%%%%%%%%%%%%%%%%%%%%%%%%%%%%%%%%%%%%%%%%%%%%%%%%%%%%%%%%%%%%%%%%%%%%%
\vspace{0.5cm}
\begin{abstract}
\noindent We obtain a dispersive long-time decay in weighted
energy norms for solutions of the 1D  Dirac equation with generic
potential. The decay extends the results obtained by Jensen,
Kato and Murata for the  Schr\"odinger equations.

\smallskip

\smallskip

\end{abstract}

\end{titlepage}
%%%%%%%%%%%%%%%%%%%%%%%%%%%%%%%%%%%%%%%%%%%%%%%%%%%%%%%%%%%%%%%%%%%%%%%%%%%%%%%%%%

%%%%%%%%%%%%%%%%%%%%%%%%%%%%%%%%%%%%%%%%%%%%%%%%%%%%%%%%%%%%%%%%%%%%%%%%%%%%%%%%%%%%%
\setcounter{equation}{0}
\section{Introduction}
%%%%%%%%%%%%%%%%%%%%%%%%%%%%%%%%%%%%%%%%%%%%%%%%%%%%%%%%%%%%%%%%%%%%%%%%%%%%%%%%%%%%%
In this paper, we establish a dispersive long time decay
for the solutions to 1D Dirac equation
\be\la{Dir}
i\dot\psi(x,t)=\cH\psi(x,t):=
i\al\psi'(x,t)+m\beta\psi(x,t)+{\cal V}(x)\psi(x,t),\quad x\in\R,\quad m> 0
\ee
in weighted energy norms.
Here   $\psi(x,t)\in\C^2$ for $(x,t)\in\R^2$,
\be\la{albe}
\al= \left(\ba{ll}
-1 & 0\\
0 & 1
\ea  \right),\quad
\beta=
\left(\ba{ll}
0 & 1\\
1& 0\\
\ea  \right)
\ee
and ${\cal V}(x)$ a given Hermitian  matrix potential: 
\be\la{Vij}
{\cal V}(x)=\left(\ba{ll}
{\cal V}_{11}(x) & {\cal V}_{12}(x)\\
{\cal V}_{21}(x) & {\cal V}_{22}(x)
\ea  \right),~~~~~~
{\cal V}_{11}(x),~{\cal V}_{22}(x)\in\R,~~~~
{\cal V}_{21}(x)=\ov {{\cal V}_{12}(x)},~~~~~~~~x\in\R.
\ee
Hence, the Dirac operator $\cH$ is Hermitian.
The matrices  $\al$ and $\beta$ satisfy the  relations
\be\la{ab}
\al^2=\beta^2=I,\quad \al\beta+\beta\al=0
\ee
For $s,\si\in\R$,
let us denote by $H^s_\si=H^s_\si (\R^3)$
the weighted Sobolev spaces \ci{A}, with the finite norms
$$
  \Vert u\Vert_{H^s_\si}=\Vert\langle x
\rangle^\si\langle\na\rangle^s u\Vert_{L^2(\R)}<\infty,
\quad\quad \langle x\rangle=(1+|x|^2)^{1/2}
$$
We assume that
\be \label{V}
|{\cal V}_{11}(x)|+|{\cal V}_{22}(x)|+
|{\cal V}_{12}(x)|+|{\cal V}'_{12}(x)|\le C\langle x\rangle^{-\beta},~~x\in\R
\ee
for some $\beta>5$. Then the multiplication by ${\cal V}_{12}$
is bounded operators  $H^1_s \to H^1_{s+\beta}$ for any $s\in\R$.

We restrict ourselves to ``nonsingular case''
when the truncated resolvent of the  operator $\cH$
is bounded at the edge points $\lam=\pm m$  of the continuous spectrum.

%%%%%%%%%%%%%%%%%%%%%%%%%%%%%%%%%%%%%%%%%%%%%%%%%%%%%
Our main result is the following long time decay of the solutions
to (\re{Dir}): in the nonsingular case
\be\label{full}
 \Vert\cP_c\psi(t)\Vert_
 {L^2_{-\si}}=\cO (|t|^{-3/2}),\quad t\to\pm\infty
\ee
for initial data $\psi_0=\psi(0)\in L^2_\si:=H^0_\si\otimes\C^2$ with $\sigma>5/2$
where  $\cP_c$ is a Riesz projector onto the continuous spectrum
of the operator $\cH$.
%%%%%%%%%%%%%%%%%%%%%%%%%%%%%%%%%%%%%%%%%%%%%%%%%%%%%%%%%%%%%%%%%%%%%%%%
The decay is desirable for the study of asymptotic stability and scattering
for solutions to nonlinear Dirac equations.

Let us comment on previous results in this direction. The decay of
type (\re{full}) in weighted norms has been established first by
Jensen and Kato \ci{jeka} for the Schr\"odinger equation in the
dimension $n=3$. The result has been extended to all other
dimensions by Jensen and Nenciu \ci{je1,je2,JN2001}, and to more
general PDEs of the Schr\"odinger type by Murata \ci{M}. In
\ci{1dkg}-\ci{3dkg} the decay of type (\re{full}) in  weighted
energy norms has been proved  for the Klein-Gordon equations. For
Dirac equations the Strichartz estimates
were established in \cite{MNNO} . The decay $\sim t^{-1}$ in
$L^\infty $ norm were established in \ci{AF} for Dirac equations
with small  potential. The decay of type (1.9) in weighted
norms for and Dirac equation without any smallness conditions on
the potential was not obtained before.
%%%%%%%%%%%%%%%%%%%%%%%%%%%%%%%%%%%%%%%%%%%%%%%%%%%%%%%%%%%%%%%%%%%%%%%

Let us comment on our techniques.
We extend our approach \ci{1dkg} to the Dirac equation. It is well known
that the decay (\re{full}) violates for the free 1D Klein-Gordon and Dirac
equation corresponding to ${\cal V}(x)=0$ when the solutions slow decay, 
like $\sim t^{-1/2}$.
Hence, the decay (\re{full}) cannot be deduced by perturbation arguments
from the corresponding estimate for the free equation.
The slow decay is caused by the ``zero resonance function''  $\psi(x)=\const$
corresponding to the end point $\lam=0$ of the continuous
spectrum of the  operator $d^2/dx^2$.

Main idea of our approach is a spectral analysis of the ``bad'' term,
with the slow decay $\sim t^{-1/2}$. Namely, we show that the bad term does not
contribute to the high energy component of solution to the free equation,
and the high energy component decays  like $t^{-3/2}$.
Then we prove the decay $\sim t^{-3/2}$ for the high energy component
of solution to perturbed equation (\re{Dir}) 
using finite Born series and convolutions. 
For the proof we apply a gauge transformation to obtain a suitable expression 
for the resolvent of the operator $\cH$ via  resolvent of the corresponding "squared
Dirac operator" which is a matrix Schr\"odinger operator  with a perturbation. 
The perturbation does not contain differential operators which allows us to apply 
the  estimates (\re{AR}) obtained in \cite{jeka} for Schr\"odinger operator.

For the low energy component of solution to perturbed equation, 
the decay $\sim t^{-3/2}$ follows 
in the ``nonsingular case'' by methods \ci{jeka, M}.
%%%%%%%%%%%%%%%%%%%%%%%%%%%%%%%%%%%%%%%%%%%%%%%%%%%%%%%%%%%%%%%%%%%%%%%%%%%%%%%

Our paper is organized as follows.
In Section \ref{fKG} we obtain the time decay for the solution to the
free Dirac
equation and state the spectral properties of the free resolvent which follow from the
corresponding known properties of the free Schr\"odinger resolvent.
In Section \ref{pKG} we obtain spectral properties of the perturbed resolvent
and prove the decay \eqref{full}.
%%%%%%%%%%%%%%%%%%%%%%%%%%%%%%%%%%%%%%%%%%%%%%%%%%%%%%%%%%%%%%%%%%%%%%%%%%%%%%%%%%
%%%%%%%%%%%%%%%%%%%%%%%%%%%%%%%%%%%%%%%%%%%%%%%%%%%%%%%%%%%%%%%%%%%%%%%%%%%%%%%%%%
\setcounter{equation}{0}
\section{Free Dirac equation}\la{fKG}
%%%%%%%%%%%%%%%%%%%%%%%%%%%%%%%%%%%%%%%%%%%%%%%%%%%%%%%%%%%%%%%%%%%%%%%%%%%%%%%%%%
%%%%%%%%%%%%%%%%%%%%%%%%%%%%%%%%%%%%%%%%%%%%%%%%%%%%%%%%%%%%%%%%%%%%%%%%%%%%%%%%%%
First, we consider the free Dirac equation:
\be\la{Dir0}
i\dot\psi(x,t)=\cH_0\psi(x,t):=i\al\psi'(x,t)+m\beta\psi(x,t)
\ee
Denote by $\cU(t):\psi(\cdot,0)\to\psi(\cdot,t)$ the dynamical  group
of the equation (\ref{Dir0}).
It is strongly continuous group in $L^2:=L^2(\R)\otimes\C^2$.
The group is unitary that follows from the charge conservation.
%%%%%%%%%%%%%%%%%%%%%%%%%%%%%%%%%%%%%%%%%%%%%%%%%%%%%%%%%%%%%%%%%%%%%%%%%%%%%%%%%%%
%%%%%%%%%%%%%%%%%%%%%%%%%%%%%%%%%%%%%%%%%%%%%%%%%%%%%%%%%%%%%%%%%%%%%%%%%%%%%%%%%%%%%
\subsection{Spectral properties}
%%%%%%%%%%%%%%%%%%%%%%%%%%%%%%%%%%%%%%%%%%%%%%%%%%%%%%%%%%%%%%%%%%%%%%%%%%%%%%%%%%%%%
%%%%%%%%%%%%%%%%%%%%%%%%%%%%%%%%%%%%%%%%%%%%%%%%%%%%%%%%%%%%%%%%%%%%%%%%%%%%%%%%%%%%%
We state spectral properties of the free  dynamical group
$\cU(t)$ applying known results of \ci{A, M} which concern the
corresponding spectral properties of the free  Schr\"odinger
dynamical group. For $t>0$ and $\psi_0=\psi(0)\in L^2$,
the solution $\psi(t)$ to the free  equation (\re{Dir0})
admits the spectral Fourier-Laplace representation
\be\la{Gint}
  \theta(t)\psi(t)=\fr 1{2\pi i}\int\limits_{\R}e^{-i(\om+i\ve) t}
  \cR _0(\om+i\ve)\psi_0~d\om,~~~~t\in\R
\ee
with any $\ve>0$ where $\theta(t)$ is the Heaviside function,
$\cR _0(\omega)=({\cal H}_0-\omega)^{-1}$ is the resolvent of the operator ${\cal H}_0$.
The representation follows from the stationary equation
$\om\ti\psi^+(\om)=\cH_0\ti\psi^+(\om)+i\psi_0$
for the Fourier-Laplace transform
$\ti\psi^+(\om):=\ds\int_\R \theta(t)e^{i\om t}\psi(t)dt$,
$\om\in\C^+:=\{\rIm\om>0\}$.
The solution  $\psi(t)$ is  continuous bounded  function of
$t\in\R$ with the values in $L^2$
by the charge conservation for the free  equation (\re{Dir0}).
Hence, $\ti\psi^+(\om)=-i\cR_0(\om)\psi_0$
is analytic function of $\om\in\C^+$ with the values in $L^2$,
bounded for $\om\in\R+i\ve$. Therefore, the integral (\re{Gint})
converges in the sense of distributions of $t\in\R$ with the values in $L^2$.
Similarly to (\re{Gint}),
\be\la{Gints}
  \theta(-t)\psi(t)=-\fr 1{2\pi i}\int\limits_{\R}e^{-i(\om-i\ve) t}
  \cR _0(\om-i\ve)\psi_0~d\om,~~~~t\in\R
\ee
The resolvent $\cR _0(\omega)$ can be expressed in terms of the resolvent
$R_0(\zeta)=(-\pa^2_x-\zeta)^{-1}$ of the free Schr\"odinger
operator. Indeed, (\re{ab}) implies
\be\la{HH}
(\cH_0-\om)(\cH_0+\om)
=(i\al\pa_x+m\beta-\om)(i\al\pa_x+m\beta+\om)=(-\pa^2_x+m^2-\om^2)
\ee
Therefore,
\be\la{R0R0}
\cR_0(\om)=(i\al\pa_x+m\beta+\om)R_0(\om^2-m^2)
\ee
where $R_0(\zeta)$ is the operator with the integral kernel
\be\la{ef}
 R_0(\zeta,x-y)=-\frac{\exp(i\sqrt\zeta|x-y|)}{2i\sqrt\zeta},
 \quad\zeta\in\C\setminus [0,\infty),\quad\rIm\zeta^{1/2}>0
\ee
%%%%%%%%%%%%%%%%%%%%%%%%%%%%%%%%%%%%%%%%%%%%%%%%%%%%%%%%%%%%%%%%%%%%%%%%%%%%%%%%%%
Denote by $\cL (B_1,B_2)$ the Banach space of bounded linear operators
from a Banach space $B_1$ to a Banach space $B_2$.
Explicit formula (\re{ef}) obviously implies the
properties of $R_0(\zeta)$ (cf. \ci{A, M}):
\medskip\\
i) $R_0(\zeta)$ is  analytic function of $\zeta\in\C\setminus [0,\infty)$
with the values in $\cL(H^0,H^2)$;
\\
ii) For $\zeta>0$, the convergence holds
\be\la{lap}
R_0(\zeta\pm i\ve)\to R_0(\zeta\pm i0), \quad\ve\to 0+
\ee
in $\cL(H^0_\si,H^2_{-\si})$ with $\si>1/2$, uniformly in $\zeta\ge r$ for any $r>0$.
\\
iii) In  $\cL (H^{0}_{\si};H^2_{-\si})$ with $\si>5/2$
the  asymptotics  holds
\be\la{exp0}
R_0(\zeta)= A_0\zeta^{-1/2}+A_1+\cO (\zeta^{1/2}),\quad
R'_0(\zeta)=-\fr 12A_0\zeta^{-3/2}+\cO (\zeta^{-1/2}),\quad
R''_0(\zeta)=\cO (\zeta^{-5/2})
\ee
as $\zeta\to 0$, $\zeta\in\C\setminus [0,\infty)$.
Here
\beqn\la{AA}
&&A_0={\rm Op}\Big[\frac i{2}\Big]\in\cL(H^{0}_{\si};H^2_{-\si}),~~\si>1/2\\
\nonumber
&&A_1={\rm Op}\Big[-\frac 12|x-y|\Big]\in\cL (H^{0}_{\si};H^2_{-\si}),~~\si>3/2
\eeqn
iv) For  $l=0,1$,  $k=0,1,2,...$ and $\si>1/2+k$  the asymptotics hold
\be\la{A}
 \Vert R_0^{(k)}(\zeta)\Vert_{\cL (H^0_\si,H^{l}_{-\si})}=\cO(|\zeta|^{-\fr{1-l+k}2}),
  \quad\zeta\to\infty,\quad \zeta\in\C\setminus(0,\infty)
\ee
\smallskip\\
%%%%%%%%%%%%%%%%%%%%%%%%%%%%%%%%%%%%%%%%%%%%%%%%%%%%%%%%%%%%%%%%%%%%%%%%%%%%%
Let us denote $\Ga:=(-\infty,-m)\cup(m,\infty)$,
The properties i) -- iv) and formula (\re{R0R0}) imply
%%%%%%%%%%%%%%%%%%%%%%%%%%%%%%%%%%%%%%%%%%%%%%%%%%%%%%%%%%%%%%%%%%%%%%%%%%%%%%%%%%%
\begin{lemma}\la{sp}
i) The resolvent $\cR _0(\om)$ is  analytic function of
$\om\in\C\setminus\ov\Ga$ with the values in  $\cL (L^2,L^2)$.
\\
ii) For $\om\in\Ga$, the convergence holds
\be\la{lap1}
\cR_0(\om\pm i\ve)\to \cR_0(\om\pm i0),~~\ve\to 0+
\ee
in $\cL(L^2_{\si},L^2_{-\si})$ with $\si>1/2$, uniformly in $|\om|\ge m+r$ for any $r>0$.
\\
iii) In $\cL (L^2_{\si};L^2_{-\si})$ with $\si>5/2$ the  asymptotics hold
 \begin{equation}\label{expR0}
\left.\ba{lll}
\cR_0(\om)=\cA_0^{\pm}(\om\mp m)^{-1/2}+\cA_1^{\pm}
+{\cal O}((\om\mp m)^{1/2})\\\\
\cR'_0(\om)=-\fr 12\cA_0^{\pm}(\om\mp m)^{-3/2}+{\cal O}((\om\mp m)^{-1/2})\\\\
\cR''_0(\om)={\cal O}((\om\mp m)^{-3/2})\ea\right|
~~\om\to\pm m,\quad\om\in\C\setminus\ov\Gamma
\end{equation}
Here
\beqn\nonumber
\cA_0^{\pm}\!\!\!&=&\!\!\!{\rm Op}\Big[\frac{i\sqrt m}{\sqrt{\pm 2}}\left(\ba{cc}
\pm 1 & 1\\
1 & \pm1
\ea  \right)\Big]\in\cL(L^2_{\si};L^2_{-\si}),~~\si>1/2\\\\
\nonumber\la{cAA}
\cA_1^{\pm}\!\!\!&=&\!\!\!{\rm Op}\Big[\frac {-m|x-y|}2\left(\!\ba{cc}
\pm 1 & 1\\
1 & \pm1
\ea\!\right)-\frac 12\left(\!\ba{cc}
\sgn (x-y) & 0\\
0 & \sgn (x-y)
\ea\!\right)\Big]\in\cL(L^2_{\si};L^2_{-\si}),~~\si>3/2
\eeqn
iv) For $k=0,1,2,...$ and  $\si>1/2+k$ the asymptotics hold
  \begin{equation}\label{bR0}
    \Vert \cR _0^{(k)}(\om) \Vert_{\cL (L^2 _{\sigma},L^2 _{-\sigma})}=\cO(1)
    \quad\om\to\infty,\quad\om\in\C\setminus\Ga
  \end{equation}
\end{lemma}
\begin{cor}
\la{irep}
For $t\in\R$ and $\psi_0\in L^2_\si$ with $\si>1/2$,
the group $\cU(t)$ admits the integral representation
\be\la{Gint1}
  \cU(t)\psi_0=\frac 1{2\pi i}\int\limits_\Gamma e^{-i\om t}
  \Big[\cR _0(\om+i0)-\cR _0(\om-i0)\Big]\psi_0~ d\om
\ee
where the integral converges in the sense of distributions of $t\in\R$
with the values in $L^2_{-\si}$.
\end{cor}
\begin{proof}
Summing up the representations  (\ref{Gint}) and  (\ref{Gints}), and sending $\ve\to 0+$,
we obtain  (\ref{Gint1}) by the Cauchy theorem and Lemma \re{sp}.
\end{proof}
%%%%%%%%%%%%%%%%%%%%%%%%%%%%%%%%%%%%%%%%%%%%%%%%%%%%%%%%%%%%%%%%%%%%%%%%%%%%%%%%%%
%%%%%%%%%%%%%%%%%%%%%%%%%%%%%%%%%%%%%%%%%%%%%%%%%%%%%%%%%%%%%%%%%%%%%%%%%%%%%%%%%%%%%
\subsection{Time decay}\la{Tdec}
%%%%%%%%%%%%%%%%%%%%%%%%%%%%%%%%%%%%%%%%%%%%%%%%%%%%%%%%%%%%%%%%%%%%%%%%%%%%%%%%%%%%%
%%%%%%%%%%%%%%%%%%%%%%%%%%%%%%%%%%%%%%%%%%%%%%%%%%%%%%%%%%%%%%%%%%%%%%%%%%%%%%%%%%%%%
Here we prove the time decay (\ref{full}) for the free Dirac
equation (\re{Dir0}).
Let $G(t)={\rm Op}[G(x-y,t)]$, where $G(x,t)$
is the fundamental solution to the Klein-Gordon operator
$\pa^2_t-\pa_x^2+m^2$:
$$
G(x,t)=\frac 12\theta(t-|x|)J_0(m\sqrt{t^2-x^2})
$$
where $J_0$ is the Bessel function.
Since
$$
(\pa_t+\al\pa_x-im\beta)(\pa_t-\al\pa_x+im\beta)=\pa^2_t-\pa^2_x+m^2
$$
then for the integral kernel $\cU(x-y,t)$ of the operator $\cU(t)$ we have
\be\la{UG}
\cU(x,t)=(\pa_t+\al\pa_x-im\beta)G(x,t)
\ee
The asymptotics of the  Bessel functions \cite{W} imply
$$
G(x,t)\sim t^{-1/2},\quad \pa_t G(x,t)\sim t^{-1/2},
\quad \pa_x G(x,t)\sim t^{-1/2}, \quad x\in\R, \quad t\to\infty
$$
Hence  (\re{UG}) implies
$$
\cU(x,t)\sim t^{-1/2},\quad x\in\R,\quad t\to\infty
$$
and then the free group $\cU(t)$ decays like $t^{-1/2}$.
The slow decay is caused by the presence of resonance at the edge
points $\zeta=\pm m$ of the continuous spectrum.
\\
In \cite{1dkg} we obtain similarly (\re{Gint1}) the integral representation
for the operator $G(t)$:
\be\la{Gint2}
G(t)\psi_0=\frac 1{2\pi }\int\limits_\Gamma e^{-i\om t}
\Big[R _0((\om+i0)^2-m^2)-R _0((\om-i0)^2-m^2)\Big]\psi_0~ d\om
\ee
for $t\in\R$ and $\psi_0\in L^2_\si$ with $\si>1/2$.
Let us introduce the   {\it low energy} and
{\it high energy} components of  $G(t)$ and $\cU(t)$:
\beqn\label{Gl}
G_l(t)&=&\frac 1{2\pi}\int\limits_\Gamma e^{-i\om t}l(\om)
  \Big[R_0((\om+i0)^2-m^2)-R_0((\om-i0)^2-m^2)\Big]~ d\om\\
\label{Gh}
G_h(t)&=&\frac 1{2\pi}\int\limits_\Gamma e^{-i\om t}h(\om)
  \Big[R_0((\om+i0)^2-m^2)-R_0((\om-i0)^2-m^2)\Big]~ d\om\\
\label{Ul}
\cU_l(t)&=&\frac 1{2\pi i}\int\limits_\Gamma e^{-i\om t}l(\om)
  \Big[\cR _0(\om+i0)-\cR _0(\om-i0)\Big]~ d\om\\
\label{Uh}
\cU_h(t)&=&\frac 1{2\pi i}\int\limits_\Gamma e^{-i\om t}h(\om)
  \Big[\cR _0(\om+i0)-\cR _0(\om-i0)\Big]~ d\om
\eeqn
where
$l(\om)\in C_0^\infty(\R)$ is an even function,
$\supp l\in [-m-2\ve, m+2\ve]$, $l(\om)=1$ if $|\om|\le m+\ve$,
and $h(\om)=1-l(\om)$.
In \cite{1dkg} we have proved that $G_h(t)$ decays like $t^{-3/2}$.
Here we will prove that $\cU_h(t)$ also decays like $t^{-3/2}$.
%%%%%%%%%%%%%%%%%%%%%%%%%%%%%%%%%%%%%%%%%%%%%%%%%%%%%%%%%%%%%%%%%%%%%%%%%%%%%%%
\begin{theorem}\la{TD}
Let $\sigma>5/2$. Then the decay holds
\begin{equation}\label{Gb1}
\Vert\cU_h(t)\Vert_{\cL(L^2_\si;L^2_{-\si})}\le C(1+|t|)^{-3/2},\quad t\in\R
\end{equation}
\end{theorem}
%%%%%%%%%%%%%%%%%%%%%%%%%%%%%%%%%%%%%%%%%%%%%%%%%%%%%%%%%%%%%%%%%%%%%%%%%%%%%%
\begin{proof}
First, taking into account (\ref{UG}), (\ref{Gh}), (\ref{Uh})
we obtain that
\be\la{GUh}
\cU_h(t)=(\pa_t+\al\pa_x-im\beta)G_h(t)
\ee
Further, \cite[Theorem 2.7]{1dkg} implies that for  $\sigma>5/2$
\be\la{G-dec}
\Vert G_h(t)\Vert_{\cL(H^0_\si;H^1_{-\si})}+
\Vert\pa_t G_h(t)\Vert_{\cL(L^2_\si;L^2_{-\si})}\le C(1+|t|)^{-3/2}),
\quad t\in\R
\ee
Then (\re{GUh}) and (\re{G-dec}) imply  (\ref{Gb1}).
\end{proof}
%%%%%%%%%%%%%%%%%%%%%%%%%%%%%%%%%%%%%%%%%%%%%%%%%%%%%%%%%%%%%%%%%%%%%%%%%%%%%%%%
%%%%%%%%%%%%%%%%%%%%%%%%%%%%%%%%%%%%%%%%%%%%%%%%%%%%%%%%%%%%%%%%%%%%%%%%%%%%%%%%
\setcounter{equation}{0}
\section{Perturbed  equation}\la{pKG}
%%%%%%%%%%%%%%%%%%%%%%%%%%%%%%%%%%%%%%%%%%%%%%%%%%%%%%%%%%%%%%%%%%%%%%%%%%%%%%%%
%%%%%%%%%%%%%%%%%%%%%%%%%%%%%%%%%%%%%%%%%%%%%%%%%%%%%%%%%%%%%%%%%%%%%%%%%%%%%%%%
To prove the long time decay for the perturbed  equation,
we first establish the spectral properties of its generator ${\cal H}$.
%%%%%%%%%%%%%%%%%%%%%%%%%%%%%%%%%%%%%%%%%%%%%%%%%%%%%%%%%%%%%%%%%%%%%%%%%%%%%%%%%%%%%%%%%%%
\subsection{Spectral properties}
%%%%%%%%%%%%%%%%%%%%%%%%%%%%%%%%%%%%%%%%%%%%%%%%%%%%%%%%%%%%%%%%%%%%%%%%%%%%%%%%%%%%%%%%%%%
Similarly to  \ci[formula (3.1)]{M},
let us introduce a generalized eigenspaces $\bf M^{\pm}$
of  the  operator $\cH$:
$$
{\bf M^{\pm}}=\{\psi\in H^1_{-1/2-0}:
\,~(1+\cA_1^{\pm}\cV)\psi\in\Re(\cA_{0}^{\pm}),\quad
\cA_{0}^{\pm}\cV\psi=0\}
$$
Where $\cA_{0}^{\pm}$ and $\cA_1^{\pm}$ are defined in (\ref{cAA}),
$\Re(\cA_{0}^{\pm})$ is the range of $\cA_{0}^{\pm}$.
Below we assume that
\be\la{SC}
 {\bf M^{\pm}}=0
\ee
Denote by $\cR(\om)=(\cH-\om)^{-1}$, $\om\in\C\setminus\Gamma$,
the resolvents of the  operators $\cH$.
Next Lemma is  the vector version of \ci[Theorem 7.2]{M}.
%%%%%%%%%%%%%%%%%%%%%%%%%%%%%%%%%%%%%%%%%%%%%%%%%%%%%%%%%%%%%%%%%%%%%%%%%%%%%%%%%%%%
\begin{lemma}\la{R-b}
Let the  conditions (\ref{V}) and (\ref{SC}) hold.
Then the families
$\{\cR(\pm m+\ve): \pm m+\ve\in\C\setminus\ov\Gamma, |\ve|<\de\}$
are bounded in the operator norm
of ${\cal L}(L^2_{\si},L^2_{-\si})$ for any $\si>3/2$ and
sufficiently small $\de$.
\end{lemma}
Asymptotics (\ref{expR0}) and lemma \re{R-b} imply
%%%%%%%%%%%%%%%%%%%%%%%%%%%%%%%%%%%%%%%%%%%%%%%%%%%%%%%%%%%%%%%%%%%%%%%%%%%%%%%%%
\begin{pro}\la{R-exp}
Let the  conditions (\ref{V}) and (\ref{SC}) hold.
Then  the asymptotics hold
\be\la{Rexp}
\left.\ba{lll}
\cR(\om)-\cR(\pm m)=\cO(|\om\mp m|^{1/2})\\
\cR'(\om)=\cO(|\om\mp m|^{-1/2})\\
\cR''(\om)=\cO(|\om\mp m|^{-3/2})\ea\right|
~~\om\to\pm m,\quad\om\in\C\setminus\ov\Gamma
\ee
in ${\cL (L^2_{\si},L^2_{-\si})}$ with $\si>5/2$.
\end{pro}
%%%%%%%%%%%%%%%%%%%%%%%%%%%%%%%%%%%%%%%%%%%%%%%%%%%%%%%%%%%%%%%%%%%%%%%%%%%%%%%%%
\begin{proof}
Lemma \re{R-b} implies that
for any $\si>3/2$  the operators
$(1+\cR_0(\om)\cV)^{-1}=1-\cR(\om)\cV$
and $(1+\cV\cR_0(\om))^{-1}=1-\cV\cR(\om)$  are bounded
in ${\cal L}(L^2_{-\si},L^2_{-\si})$ and in
${\cal L}(L^2_{\si},L^2_{\si})$ respectively
for $|\om\mp m|<\de$, $\om\in\C\setminus\ov\Ga$
with $\de$ sufficiently small.
Asymptotics (\ref{expR0}) imply
$$
\!\left.\ba{ll}
\cR(\om)=\big(1+\cR_0(\om)\cV\big)^{-1}\cR_0(\om)
=\big(1+\cR_0(\om)\cV\big)^{-1}\big(\cA_0^{\pm}
\ds\frac{1}{\sqrt{\om\mp m}}+{\cal O}(1)\big)\\\\
\cR(\om)=\cR_0(\om)\big(1+\cV\cR_0(\om)\big)^{-1}
=\big(A_0^{\pm}\ds\frac{1}{\sqrt{\om\mp m}}+{\cal O}(1)\big)
\big(1+\cV\cR_0(\om)\big)^{-1}\ea\!\right|
~\om\!\to\pm m, ~~\om\in\C\setminus\ov\Ga
$$
in ${\cal L}(L^2_{\si}, L^2_{-\si})$ with $\si>3/2$.
Hence, the boundedness $\cR(\om)$, $(1+\cR_0(\om)\cV)^{-1}$
and $(1+\cV\cR_0(\om))^{-1}$ at the points $\om=\pm m$
in corresponding norms imply that
\be\la{A00}
\!(1+\cR_0(\om)\cV)^{-1}\cA_0^{\pm}\!={\cal O}(\sqrt{\om\mp m}),\quad
\cA_0^{\pm}(1+\cV\cR_0(\om))^{-1}\!={\cal O}(\sqrt{\om\mp m}),~~\om\!\to 0,
~~\om\in\C\setminus\ov\Ga
\ee
in ${\cal L}(L^2_{\si},L^2_{-\si})$ with $\si>3/2$.
%%%%%%%%%%%%%%%%%%%%%%%%%%%%%%%%%%%%%%%%%%%%%%%%%%%%%%%%%%%%%%%%%%%%%%%%%%%%%%%%%%%
Therefore,
\be\la{G00}
\Vert(1+\cR_0(\om)V)^{-1}[1]\Vert_{L^2_{-\si}}={\cal O}(\sqrt{\om\mp m}),
\quad\om\to\pm m, \quad\om\in\C\setminus\ov\Ga,\quad\si>3/2
\ee
and  for any $f\in L^2_{\si}$ with $\si>3/2$
\be\la{G01}
\int [(1+\cV\cR_0(\om))^{-1}f](x)dx={\cal O}(\sqrt{\om\mp m}),
\quad\om\to\pm m, \quad\om\in\C\setminus\ov\Ga,\quad\si>3/2
\ee
Taking into account the identities
$$
\cR'=(1+\cR_0\cV)^{-1}\cR'_0(1+\cV\cR_0)^{-1},\quad
\cR''=\Big[(1+\cR_0\cV)^{-1}\cR''_0
-2\cR'\cV\cR'_0\Big](1+\cV\cR_0)^{-1}
$$
we obtain from  (\ref{G00})-(\ref{G01}) the asymptotics
(\re{Rexp}) for  $\cR'(\om)$ and  $\cR''(\om)$.
Finally, the asymptotics (\re{Rexp})  for $\cR(\om)$ follow by
integration the asymptotics  (\re{Rexp}) for $\cR'(\om)$.
\end{proof}
%%%%%%%%%%%%%%%%%%%%%%%%%%%%%%%%%%%%%%%%%%%%%%%%%%%%%%%%%%%%%%%%%%%%%%%%%%%%%%%%%%%%%%
To obtain other properties of $\cR(\om)$ we express $\cR(\om)$ in
the resolvent of a matrix Schr\"odinger operator.
First, we introduce the operator
$\tilde\cH$ similar to $\cH$ with matrix potential $\tilde\cV(x)$
satisfying $\tilde\cV_{11}(x)=\tilde\cV_{22}(x)=0$.
Namely, let us introduce the matrix of gauge transformation
\be\la{C}
C=\left(\begin{array}{cc}
 \exp\big(\!-i\!\!\int\limits_{-\infty}^x\!\!{\cal V}_{11}(y)dy\big) &   0
  \\
  0   &   \exp\big(i\!\!\int\limits_{-\infty}^x\!\!{\cal V}_{22}(y)dy\big)
\end{array}\right)
\ee
This matrix function is bounded by the  conditions (\re{Vij}), and
\be\la{CHC}
C^{-1}\cH C=\tilde\cH
\ee
where
\be\la{tH}
\tilde\cH:=
i\al\pa_x+\tilde {\cal V},\quad\tilde {\cal V}=
\left(\begin{array}{cc}
      0         &    \tilde {\cal V}_{12}
  \\
 \tilde {\cal V}_{21}             &   0
\end{array}\right)
\ee
with
\be\la{tV}
\tilde {\cal V}_{12}(x)=
\exp\Big(i\!\!\int\limits_{-\infty}^x\!\!({\cal V}_{11}(y)
+{\cal V}_{22}(y))dy\Big)({\cal V}_{12}(x)+m),\quad
\tilde {\cal V}_{21}(x)=\ov{\tilde {\cal V}_{12}(x)}
\ee
where we have used the conditions (\re{Vij}).
By (\re{CHC}) the spectral properties of $\cH$ and $\tilde\cH$
are identical. Further, we have
\beqn\nonumber
&&\!\!\!\!\!\!\!(\tilde\cH-\omega)(\tilde\cH+\omega)
=(i\al\pa_x+\tilde \cV(x)-\om)(i\al\pa_x+\tilde {\cal V}(x)+\om)\\
\la{ID}\\
\nonumber
&&\!\!\!\!\!\!\!\!\!\!=\!\left(\!\!\ba{cc}
-\pa^2_x+m^2\!-\om^2\!+2m\rRe{\cal V}_{12}+|{\cal V}_{12}|^2
&  -i\tilde {\cal V}_{12}'\\
i\tilde {\cal V}_{21}'
& -\pa^2_x+m^2\!-\om^2\!+2m\rRe{\cal V}_{12}+|{\cal V}_{12}|^2\ea\!\!\right)
\!={\bf H}-\!(\om^2\!-m^2)
\eeqn
Here  ${\bf H}={\bf H}_0+V$ is the matrix  Schr\"odinger operator with
\be\la{qq}
{\bf H}_0=\left(\ba{cc}
-\pa^2_x &   0\\
0        & -\pa^2_x
\ea\right),\quad
V=\left(\ba{cc}
2m\rRe{\cal V}_{12}+|{\cal V}_{12}|^2 &  -i\tilde {\cal V}_{12}'\\
i\tilde {\cal V}_{21}' & 2m\rRe{\cal V}_{12}+|{\cal V}_{12}|^2
\ea\right)
\ee
\begin{remark}\la{do}
Due (\re{Vij}), 
the  perturbation $V$ is Hermitian and  does not contain 
differential operators. This fact is essential for obtaining high energy decay
 (\re{vp1}), (\re{uv}), and (\re{lins3}).
\end{remark}
Formula (\re{ID}) implies the following representation for the resolvent
$\tilde\cR (\omega)=(\tilde\cH-\omega)^{-1}$:
\be\la{calR}
\tilde\cR (\omega)=(\tilde\cH+\omega)R(\om^2-m^2)=
(i\al\pa_x+\tilde {\cal V}(x)+\om)R(\om^2-m^2)
\ee
where $R(\zeta)=({\bf H}-\zeta)^{-1}$, $\zeta\in\C\setminus[0,\infty)$,
the resolvent of the  operators ${\bf H}$.

Due to (\re {CHC}) we have $\tilde\cR (\omega)=C\cR (\omega)C^{-1}$.
The resolvent $\tilde\cR (\omega)$
admits the low energy asymptotics of type (\re{Rexp}) since
 matrix function (\re{C}) is bounded. Namely
\be\la{tRexp}
\left.\ba{lll}
\tilde\cR(\om)-\tilde\cR(\pm m)=\cO(|\om\mp m|^{1/2})\\
\tilde\cR'(\om)=\cO(|\om\mp m|^{-1/2})\\
\tilde\cR''(\om)=\cO(|\om\mp m|^{-3/2})\ea\right|
~~\om\to\pm m,\quad\om\in\C\setminus\ov\Gamma
\ee
in ${\cL (L^2_{\si},L^2_{-\si})}$ with $\si>5/2$.

Below we need the limiting absorption  principle and 
high energy decay for the resolvent $\ti\cR(\om)$.
First, we
obtain these properties for the resolvent $R(\zeta)$ of the matrix Schr\"odinger 
operator ${\bf H}$.
%%%%%%%%%%%%%%%%%%%%%%%%%%%%%%%%%%%%%%%%%%%%%%%%%%%%%%%%%%%%%%%%%%%%%%%%%%%%%
\begin{lemma}\la{pR}
Let the  conditions (\ref{V}) and (\ref{SC}) hold.  Then\\
i) $R(\zeta)$ is  meromorphic function of $\zeta\in\C\setminus[0,\infty)$.
\\
ii) For $\zeta>0$, the convergence (limiting absorption principle) holds
\be\la{LAP2}
R(\zeta\pm i\ve)\to R(\zeta\pm i0),\quad\ve\to 0+
\ee
in $\cL(H^{0}_\si,H^2_{-\si})$ with $\si>1/2$.
\\
iii) For $l=0,1$, the asymptotics hold
\be\la{AR}
\Vert R^{(k)}_j(\zeta)\Vert_{\cL (H^0_\si,H^{l}_{-\si})}
=\cO(|\zeta|^{-\fr{1-l+k}2}),\quad\zeta\to\infty,\quad
\zeta\in\C\setminus[0,\infty)
\ee
with $\si>1/2+k$ for  $k=0,1,2$.
\end{lemma}
%%%%%%%%%%%%%%%%%%%%%%%%%%%%%%%%%%%%%%%%%%%%%%%%%%%%%%%%%%%%%%%%%%%%%%%%%%%%%%
\begin{proof}
{\it Step i)} The statement i) follows from Lemma \re{sp}-i), the Born
splitting
\be\la{Born}
R(\zeta)=R_0(\zeta)(1+VR_0(\zeta))^{-1}
\ee
and the  Gohberg-Bleher theorem \cite{B, GK}
since the operators $VR_0(\zeta)$ is a compact operator in $L^2$
for  $\zeta\in\C\setminus [0,\infty)$.\\
{\it Step ii)}
The convergence (\re{LAP2}) follows from the  vector version of
Agmon's theorem \cite[Theorem 3.3 and Lemma 4.2]{A} taking into
account the absence of embedded eigenvalues follows from
the theory of ordinary differential equations
since  $V\in L^1$.\\
{\it Step iii)}
The asymptotics (\re{AR}) follows from (\re{A}) by
the vector version of \cite[Theorem 9.2]{jeka}.
\end{proof}
%%%%%%%%%%%%%%%%%%%%%%%%%%%%%%%%%%%%%%%%%%%%%%%%%%%%%%%%%%%%%%%%%%%%%%%%%%%%
Lemma \re{pR} and formula (\ref{calR}) imply 
%%%%%%%%%%%%%%%%%%%%%%%%%%%%%%%%%%%%%%%%%%%%%%%%%%%%%%%%%%%%%%%%%%%%%%%%%%%%%%%%%%%%%%%%%
\begin{lemma}\la{sp1}
Let  the  conditions (\ref{V}) and (\ref{SC}) hold.
Then
\\
i) $\tilde\cR(\om)$ is  meromorphic function of
$\om\in\C\setminus\ov\Ga$ with the values in $\cL(L^2,L^2)$;
\\
ii) For $\om\in\Ga$, the convergence (limiting absorption principle) holds
\be\la{tlap}
\tilde\cR(\om\pm i\ve)\to \tilde\cR(\om\pm i0),\quad\ve\to 0+
\ee
in $\cL(L^2_\si,L^2_{-\si})$ with $\si>1/2$;
\\
iii) For  $k=0,1,2$ and  $\si>1/2+k$ the asymptotics hold
\be\la{bR}
 \Vert\tilde\cR^{(k)}(\om)\Vert_{\cL (L^2_\si,L^2_{-\si})}=\cO(1),
 \quad |\om|\to\infty,\quad \om\in\C\setminus\Gamma
\ee
\end{lemma}
\begin{remark}
The reduction to the Schr\"odinger
operator ${\bf H}$  allow us to  deduce Lemma \re{sp1} from Lemma \re{pR}
applying known results of Agmon \cite{A} and Jensen-Kato \cite{jeka}.
\end{remark}
%%%%%%%%%%%%%%%%%%%%%%%%%%%%%%%%%%%%%%%%%%%%%%%%%%%%%%%%%%%%%%%%%%%%%%%%%%%%%%%%%%
%%%%%%%%%%%%%%%%%%%%%%%%%%%%%%%%%%%%%%%%%%%%%%%%%%%%%%%%%%%%%%%%%%%%%%%%%%%%%%%%%%
\subsection{Time decay}
%%%%%%%%%%%%%%%%%%%%%%%%%%%%%%%%%%%%%%%%%%%%%%%%%%%%%%%%%%%%%%%%%%%%%%%%%%%%%%%%%%
%%%%%%%%%%%%%%%%%%%%%%%%%%%%%%%%%%%%%%%%%%%%%%%%%%%%%%%%%%%%%%%%%%%%%%%%%%%%%%%%%%
In this section we combine the spectral properties of the perturbed
resolvent $\cR (\om)$ and time decay for the unperturbed dynamics
using the (finite) Born perturbation series.
Denote by $\Si$ the set of eigenvalues of the operators $\cH$.
Lemma \re {R-b}) imply that the eigenvalues  cannot accumulate to the points
$\pm m$ and then the set  $\Si$ is finite.
Our main result is the following.
\begin{theorem}\la{main}
Let conditions (\ref{V}) and (\ref{SC}) hold. Then
\begin{equation}\la{full2}
   \Vert e^{-it{\cal H}}-\sum\limits_{\om_j\in\Si}
   e^{-i\omega_j t}P_j\Vert_{\cL (L^2_\si,L^2_{-\si})}
  ={\cal O}(|t|^{-3/2}),\quad t\to\pm\infty
\end{equation}
with  $\sigma>\fr 52$, where  $P_j$ are the Riesz
projectors onto the corresponding eigenspaces of operator $\cH$.
\end{theorem}
\begin{proof}
By (\re{CHC}) it suffices to prove  that for  $\sigma>5/2$
\begin{equation}\la{full1}
   \Vert e^{-it{\tilde\cH}}-\sum\limits_{\om_j\in\Si}
   e^{-i\omega_j t}\tilde P_j\Vert_{\cL (L^2_\si,L^2_{-\si})}
  ={\cal O}(|t|^{-3/2}),\quad t\to\pm\infty
\end{equation}
where  $\tilde P_j=C^{-1}P_jC$ are the Riesz
projectors onto the corresponding eigenspaces of operator $\tilde H$.
Lemma \re{sp1}  and asymptotics (\re{tRexp}) imply similarly to (\re{Gint1}), that
\be\la{Psi}
 e^{-it{\tilde{\cal H}}}\psi_0
-\sum\limits_{\om_j\in\Sigma}e^{-i\omega_j t}\tilde P_j\psi_0=
\fr 1{2\pi i}\int\limits_\Gamma e^{-i\om t}
\Big[\tilde\cR (\om+i0)-\tilde\cR (\om-i0)\Big]\psi_0~ d\om
=\psi_l(t)+\psi_h(t)
\ee
where
$$
\tilde P_j\psi_0:=-\fr 1{2\pi i}\int_{|\om-\om_j|=\delta}\tilde\cR(\om)\psi_0 d\om
$$
with a small $\delta>0$, and
\be\la{idl}
\psi_l(t)=\fr 1{2\pi i}\int\limits_\Gamma l(\om)e^{-i\om t}
\Big[\tilde\cR (\om+i0)-\tilde\cR (\om-i0)\Big]\psi_0~ d\om
\ee
\be\la{idh}
\psi_h(t)=\fr 1{2\pi i}\int\limits_\Gamma h(\om)e^{-i\om t}
\Big[\tilde\cR (\om+i0)-\tilde\cR (\om-i0)\Big]\psi_0~ d\om
\ee
where $l(\om)$ and $h(\om)$ are defined in Section \ref{Tdec}.
Further we analyze $\psi_l(t)$ and $\psi_h(t)$ separately.
%%%%%%%%%%%%%%%%%%%%%%%%%%%%%%%%%%%%%%%%%%%%%%%%%%%%%%%%%%%%%%%%%%%%%%%%%
%%%%%%%%%%%%%%%%%%%%%%%%%%%%%%%%%%%%%%%%%%%%%%%%%%%%%%%%%%%%%%%%%%%%%%%%%%
%%%%%%%%%%%%%%%%%%%%%%%%%%%%%%%%%%%%%%%%%%%%%%%%%%%%%%%%%%%%%%%%%%%%%%%%%
\subsubsection{Low energy decay}
%%%%%%%%%%%%%%%%%%%%%%%%%%%%%%%%%%%%%%%%%%%%%%%%%%%%%%%%%%%%%%%%%%%%%%%%%%
%%%%%%%%%%%%%%%%%%%%%%%%%%%%%%%%%%%%%%%%%%%%%%%%%%%%%%%%%%%%%%%%%%%%%%%%%%
We consider only the integral over $(m,m+2\ve)$.
The integral over $(-m-2\ve,-m)$ is dealt with in the same way.
We prove the desired decay of $\psi_l(t)$ using
a special case of Lemma 10.2 from \cite{jeka}.
Denote by ${\bf B}$  a Banach space with the norm $\Vert\cdot\Vert\,.$
\begin{lemma}\label{jk}
Let $F\in C([m, a],\, {\bf B})$, satisfy
\be\la{Zc}
  F(m)=F(a)=0,~~~~
  \Vert F''(\om)\Vert=\cO (|\om-m|^{-3/2}),~~~\om\to m
\ee
Then
\be\la{Zyg}
  \int\limits_m^a e^{-it\omega}F(\omega)d\omega =\cO (t^{-3/2}),
  \quad t\to\infty
\ee
\end{lemma}
Due  to (\ref{tRexp}), we can apply Lemma \ref{jk} with
$F=l(\om)\big(\tilde\cR (\om+i0)-\tilde\cR (\om-i0)\big)$,
${\bf B}= \cL (L^2_{\si},L^2 _{-\si})$,
$a=m+2\ve$ with a small $\ve>0$ and $\si>5/2$, to get
\be\la{small}
  \Vert \psi_l(t)\Vert_{L^2_{-\si}}
 \le C(1+|t|)^{-3/2}\Vert \psi_0\Vert_{L^2 _\si}
  \quad t\in\R,\quad \sigma>5/2
\ee
%%%%%%%%%%%%%%%%%%%%%%%%%%%%%%%%%%%%%%%%%%%%%%%%%%%%%%%%%%%%%%%%%%%%%%%%%%
%%%%%%%%%%%%%%%%%%%%%%%%%%%%%%%%%%%%%%%%%%%%%%%%%%%%%%%%%%%%%%%%%%%%%%%%%
\subsubsection{High energy decay}
%%%%%%%%%%%%%%%%%%%%%%%%%%%%%%%%%%%%%%%%%%%%%%%%%%%%%%%%%%%%%%%%%%%%%%%%%%
%%%%%%%%%%%%%%%%%%%%%%%%%%%%%%%%%%%%%%%%%%%%%%%%%%%%%%%%%%%%%%%%%%%%%%%%%%
The resolvents $R(\zeta)$, $R_0(\zeta)$
are related by the Born perturbation series
\be\la{id}
R(\zeta)=R_0(\zeta)I-R_0(\zeta)V R_0(\zeta)
+R_0(\zeta)VR_0(\zeta)VR(\zeta), \quad\zeta\in\C\setminus[0,\infty)
\ee
which follows by iteration of
$ R(\zeta)=R_0(\zeta)I-R_0(\zeta)VR(\zeta)$.
Then by (\re{calR}) we have
\beqn\la{Id}
\tilde\cR(\om)&=&(i\al\pa_x+\tilde {\cal V}(x)+\om)
\Big[R_0(\om^2-m^2)I-R_0(\om^2-m^2)V R_0(\om^2-m^2)\\
\nonumber
&+&R_0(\om^2-m^2)V R_0(\om^2-m^2)VR(\om^2-m^2)\Big]
\eeqn
Let us substitute the series (\re{Id}) into the spectral representation
(\re{idh}) for $\psi_h(t)$:
\beqn\nonumber
&&\!\!\!\!\!\!\!\!\!\!\!\psi_h(t)
=\fr 1{2\pi i}\int\limits_\Gamma e^{-i\om t}h(\om)
(i\al\pa_x+\tilde\cV(x)+\om)
\Big[R _0(\zeta_+)-R _0(\zeta_-)\Big]\psi_0~ d\om\\
\nonumber
\!\!\!&+&\!\!\!\fr 1{2\pi i}\int\limits_\Gamma\!\! e^{-i\om t}h(\om)
(i\al\pa_x+\tilde\cV(x)+\om)\Big[R_0(\zeta_+)V R_0(\zeta_+)
-R_0(\zeta_-)V R_0(\zeta_-)\Big]\psi_0 d\om\\
\nonumber
\!\!\!&+&\!\!\!\frac 1{2\pi i}\int\limits_\Gamma e^{-i\om t}h(\om)
(i\al\pa_x+\tilde\cV(x)+\om)
\Big[R_0(\zeta_+\!)VR_0(\zeta_+\!)VR(\zeta_+\!)
-R_0(\zeta_-\!)VR_0(\zeta_-\!)VR(\zeta_-\!)\Big]\psi_0 d\om\\
\nonumber
\!\!\!&=&\!\!\!\psi_{h1}(t)+\psi_{h2}(t)+\psi_{h3}(t),\quad t\in\R
\eeqn
where
$\zeta_+=(\om+i0)^2-m^2$, $\zeta_-=(\om-i0)^2-m^2$.
We analyze each term $\psi_{hk}$, $k=1,2,3$ separately.
\\
%%%%%%%%%%%%%%%%%%%%%%%%%%%%%%%%%%%%%%%%%%%%%%%%%%%%%%%%%%%%%%%%%%%%%%%%%%%%
{\it Step i)}
By (\ref{R0R0}), (\ref{Gh}) and  (\ref{Uh}) for the first term
$\psi_{h1}(t)$ we have
$$
\psi_{h1}(t)=-i(\tilde \cV(x)-m\beta)\cG_h(t)\psi_0+\cU_h(t)\psi_0
$$
Hence, (\re{Gb1}) and  (\re{G-dec}) imply that
\be\la{lins1}
  \Vert \psi_{h1}(t)\Vert_{L^2_{-\si}}
 \le C(1+|t|)^{-3/2}\Vert \psi_0\Vert_{L^2_\si},
  \quad t\in\R,\quad \sigma>5/2
\ee
%%%%%%%%%%%%%%%%%%%%%%%%%%%%%%%%%%%%%%%%%%%%%%%%%%%%%%%%%%%%%%%%%%%%%%%%%%%%
{\it Step ii)}
Let us represent the  second term $\psi_{h2}(t)$ as
\be\la{phi}
\psi_{h2}(t)=\varphi_1(t)+\varphi_2(t)
\ee
where
$$
\varphi_1(t)=\fr1{2\pi i}\!\int\limits_\Gamma e^{-i\om t}h(\om)(\tilde\cV-m\beta)
\Big[R_0(\zeta_+)V R_0(\zeta_+)-R_0(\zeta_-)V R_0(\zeta_-)\Big]\psi_0 d\om
$$
$$
\varphi_2(t)=\fr 1{2\pi i}\int\limits_\Gamma e^{-i\om t}h(\om)(i\al\pa_x+m\beta+\om)
\Big[R_0(\zeta_+)V R _0(\zeta_+)
-R_0(\zeta_-)V R _0(\zeta_-)\Big]\psi_0 d\om
$$
$$
=\fr 1{2\pi i}\int\limits_\Gamma e^{-i\om t}h(\om)
\Big[\cR_0(\om+i0)V R _0(\zeta_+)
-\cR_0(\om-i0)V R _0(\zeta_-)\Big]\psi_0 d\om
$$
Let us consider the first term $\varphi_1(t)$.
Denote
$$
P(\om)=h(\om)(\tilde\cV\!-m\beta)\Big[R_0(\zeta_+)V R_0(\zeta_+)-
R_0(\zeta_-)V R_0(\zeta_-)\Big]\psi_0
$$
We have
$$
\supp P(\om)\in\Gamma_\ve:=(-\infty,-m-\ve)\cup(m+\ve,\infty)
$$
and $P''\in L^1(\Gamma_\ve;\cL (L^2_{\si},L^2_{-\si}))$
with $\si>5/2$  by (\ref{A}) with $l=0$ and $k=2$
since $V$ does not contain differential operators, see Remark \re{do}.
Then, two times partial integration implies that
\be\la{vp1}
\Vert \varphi_1(t)\Vert_{L^2_{-\si}}
 \le C(1+|t|)^{-2}\Vert \psi_0\Vert_{L^2_\si},\quad t\in\R,\quad \sigma>5/2
\ee
Now we consider the second term $\varphi_2(t)$.
Denote  $h_1(\om)=\sqrt{h(\om)}$
(we can assume that $h(\om)\ge 0$ and $h_1\in\C_0^\infty(\R)$).
We set
$$
\phi_{h1}=\fr 1{2\pi i}\int\limits_\Gamma e^{-i\om t}h_1(\om)
\Big[R_0(\zeta_+)-R_0(\zeta_-)\Big]\psi_0~ d\om
=-i\cG_{h_1}(t)\psi_0
$$
It is obvious that for $\phi_{h1}$ the inequality (\ref{lins1}) also holds.
Namely,
$$
  \Vert \phi_{h1}(t)\Vert_{L^2_{-\si}}
 \le C(1+|t|)^{-3/2}\Vert \psi_0\Vert_{L^2_\si},
  \quad t\in\R,\quad \sigma>5/2
$$
Now the term $\varphi_2(t)$ can be rewritten as a convolution.
%%%%%%%%%%%%%%%%%%%%%%%%%%%%%%%%%%%%%%%%%%%%%%%%%%%%%%%%%%%%%%%%%%%%%%%%%%%
\begin{lemma}(cf.  \ci[Lemma 3.11]{1dkg})
The convolution representation holds
\be\la{P2}
  \varphi_2(t)=
  i\int\limits_0^t \cU_{h_1}(t-\tau)V\phi_{h1}(\tau)~d\tau,~~~~t\in\R
\ee
where the integral converges in $L^2_{-\si}$ with $\si>5/2$.
\end{lemma}
%%%%%%%%%%%%%%%%%%%%%%%%%%%%%%%%%%%%%%%%%%%%%%%%%%%%%%%%%%%%%%%%%%%%%%%%%
Applying Theorem \ref{TD} with $h_1$ instead of $h$
to the integrand in (\re{P2}), we obtain that
\be\la{uv}
  \Vert \cU_{h_1}(t-\tau)V\phi_{h1}(\tau)\Vert_{L^2_{-\si}}
 \le\ds\fr{C\Vert V\phi_{h1}(\tau)\Vert_{L^2_{\si'}}}{(1+|t\!-\!\tau|)^{3/2}}
 \le\ds\fr{C_1\Vert\phi_{h1}(\tau)\Vert_{L^2_{\si'-\beta}}}{(1+|t\!-\!\tau|)^{3/2}}
 \le\ds\fr{C_2\Vert\psi_0\Vert_{L^2_{\si}}}{(1\!+|t\!-\!\tau|)^{3/2}(1+|\tau|)^{3/2}}
\ee
where $\si'\in (5/2,\beta-5/2)$.
Therefore, integrating here in $\tau$, we obtain by (\re{P2}) that
\be\la{lins2}
  \Vert\varphi_2(t)\Vert_{L^2_{-\si}}\le
  C(1+|t|)^{-3/2}\Vert\psi_0\Vert_{L^2 _\si}, \quad t\in\R,\quad \sigma>5/2
\ee
%%%%%%%%%%%%%%%%%%%%%%%%%%%%%%%%%%%%%%%%%%%%%%%%%%%%%%%%%%%%%%%%%%%%%%%%%%%%%%
{\it Step iii)}
Denote by $Q(\om)$ the integrand in $\psi_{h3}$.
Since $Q''\in L^1(\Gamma_\ve;\cL (L^2_{\si},L^2_{-\si}))$
with $\si>5/2$  by (\ref{A}) and (\ref{AR})
then, two times partial integration implies that
\be\la{lins3}
\Vert\psi_{h3}(t)\Vert_{L^2_{-\si}}
 \le C(1+|t|)^{-2}\Vert \psi_0\Vert_{L^2_\si},\quad t\in\R,\quad \sigma>5/2
\ee
Finally, the bounds (\re{lins1}), (\re{vp1}), (\re{lins2}) and \re{lins3}) imply
$$
\Vert\psi_h(t)\Vert_{L^2_{-\si}}\le
  C(1+|t|)^{-3/2}\Vert\psi_0\Vert_{L^2 _\si}, \quad t\in\R,\quad \sigma>5/2
$$
Theorem \re{main} is proved.
\end{proof}
\bc
The asymptotics (\re{full1}) imply (\re{full}) with the projector
\be\la{Pr}
\cP_c:=1-\cP_d,\quad \cP_d=\sum_{\om_j\in\Si}P_j
\ee
\ec
%%%%%%%%%%%%%%%%%%%%%%%%%%%%%%%%%%%%%%%%%%%%%%%%%%%%%%%%%%%%%%%%%%%%%%%%%%%%%%%%
%%%%%%%%%%%%%%%%%%%%%%%%%%%%%%%%%%%%%%%%%%%%%%%%%%%%%%%%%%%%%%%%%%%%%%%%%%%%%%%%
%%%%%%%%%%%%%%%%%%%%%%%%%%%%%%%%%%%%%%%%%%%%%%%%%%%%%%%%%%%%%%%%%%%%%%%%%%%%%%%%
\setcounter{equation}{0}
\section{Application to the asymptotic completeness}
%%%%%%%%%%%%%%%%%%%%%%%%%%%%%%%%%%%%%%%%%%%%%%%%%%
%%%%%%%%%%%%%%%%%%%%%%%%%%%%%%%%%%%
We apply the obtained results to prove the asymptotic completeness
by standard Cook's argument.
\begin{theorem}\la{sc}
Let conditions  (\re{V}) and (\re{SC}) hold.
Then
\\
i)  For solution to (\ref{Dir}) with any initial function
$\psi(0)\in L^2$, the  long time asymptotics  hold
\be\la{scat}
  \psi(t)=
  \sum\limits_{\om_j\in\Si}  e^{-i\om_j t}\psi_j
  +\cU(t)\phi_\pm+r_\pm(t)
\ee
where $\psi_j$ are the corresponding eigenfunctions,
$\phi_\pm\in L^2$ are the scattering states, and
\be\la{rem0}
   \Vert r_\pm(t)\Vert_{L^2}\to 0,~~~~~~t\to\pm\infty
\ee
ii) Furthermore,
\be\la{rem}
   \Vert r_\pm(t)\Vert_{L^2} =\cO (|t|^{-1/2})
\ee
if $\psi(0)\in L^2_\si$ with $\si>5/2$.
\end{theorem}
\begin{proof} Denote ${\cal X}_d:=\cP_d L^2$,  ${\cal X}_c:=\cP_c L^2$.
For $\psi(0)\in{\cal X}_d$ the asymptotics  (\ref{scat}) obviously hold
with $\phi_\pm=0$ and $r_\pm(t)=0$.
Hence, it remains to prove for $\psi(0)\in {\cal X}_c$ the asymptotics
\be\la{scat1}
\psi(t)=\cU(t)\phi_\pm+r_\pm(t)
\ee
with the remainder satisfying (\re{rem0}). Moreover, it suffices to prove
the asymptotics (\ref{scat1}), (\re{rem}) for $\psi(0)\in{\cal X}_c\cap L^2_\si$
with $\si>5/2$ since the space $L^2_\si$ is dense in $L^2$, while
the group $\cU(t)$ is unitary in $L^2$.
In this case Theorem \re{main} implies the decay
\be\label{fullp}
 \Vert \psi(t)\Vert_{L^2_{-\si}}\le C(1+|t|)^{-3/2}\Vert\psi(0)\Vert_{L^2_{\si}},
\quad t\to\pm\infty
\ee
The function $\psi(t)$ satisfies the equation (\re{Dir}),
$$
i\dot\psi(t)=(\cH_0+\cV) \psi(t)
$$
Hence, the corresponding Duhamel equation reads
\be\la{Dug}
   \psi(t)= \cU(t)\psi(0)+
   \int\limits_0^t \cU(t-\tau)\cV\psi(\tau)d\tau, ~~~~t\in\R
\ee
Let us rewrite  \eqref{Dug} as
\be\la{Dug1}
   \psi(t)=\cU(t)\Big[\psi(0)+\int\limits_0^{\pm\infty}
   \cU (-\tau)\cV\psi(\tau)d\tau\Big]
   -\int\limits_t^{\pm\infty} \cU(t-\tau)\cV\psi(\tau)d\tau
    =\cU(t)\phi_\pm+r_\pm(t)
\ee
It remains to prove that $\phi_\pm\in L^2$ and (\ref{rem}) holds.
Let us consider the sign ``+'' for the concreteness.
The  unitarity of $\cU (t)$   in $L^2$, the condition (\re{V})
and the decay (\ref{fullp}) imply that for $\si'\in(5/2,\min\{\si,\beta\})$
$$
   \int\limits_0^{\infty}\Vert \cU(-\tau)\cV
   \psi(\tau)\Vert_{L^2}~d\tau \le C\int\limits_0^{\infty}\Vert \cV
   \psi(\tau)\Vert_{L^2}~d\tau\le C_1\int\limits_0^{\infty}\Vert
   \psi(\tau)\Vert_{L^2_{-\si'}}d\tau
   \le C_2\int\limits_0^{\infty}\fr{\Vert\psi(0)\Vert_{L^2_{\si}}d\tau}{(1+\tau)^{-3/2}}
   <\infty
$$
since $|{\cal V}(x)|\le C\langle x\rangle^{-\beta}
\le C'\langle x\rangle^{-\si'}$.
Hence, $\phi_+\in L^2$. The estimate (\re{rem}) follows similarly.
\end{proof}
%%%%%%%%%%%%%%%%%%%%%%%%%%%%%%%%%%%%%%%%%%%%%%%%%%%%%%%%%%%%%%%%%%%%%%%%%%%%%%%%%%%

\end{document}